%% file: main.tex
\RequirePackage{tikz}

\pdfoutput=1

\documentclass[12pt]{article}
\usepackage[utf8]{inputenc}
\usepackage[version=4]{mhchem}
\usepackage{amsmath, amssymb, amsthm}
\usepackage{thmtools}
\usepackage{enumitem}
\usepackage{array, adjustbox}
\usepackage[colorlinks=true,linkcolor=blue,citecolor=blue,urlcolor=blue]{hyperref}
\usepackage{algorithm} 
\usepackage[margin=3cm,a4paper]{geometry} 
\usepackage{changepage}
\usepackage{tikzit}
\input{styples.tikzstyles}

\theoremstyle{plain} 
\newtheorem{theorem}{Theorem}

\newtheorem{proposition}[theorem]{Proposition}

\theoremstyle{definition} 
\newtheorem{definition}[theorem]{Definition}

\newtheorem{example}[theorem]{Example}

\theoremstyle{remark} 
\newtheorem{remark}[theorem]{Remark}

\newcommand{\sset}[1]{\ensuremath{\mathbb{#1}}} 
\newcommand{\rr}{\sset{R}} 

\newcommand{\bx}{\ensuremath{\mathbf{x}}}
\newcommand{\by}{\ensuremath{\mathbf{y}}}

\newcommand{\bw}{\ensuremath{\mathbf{w}}}

\renewcommand{\bf}{\ensuremath{\mathbf{f}}}
\newcommand{\bg}{\ensuremath{\mathbf{g}}}

\DeclareMathOperator{\rank}{rank} 

\newcommand{\goth}[1]{\ensuremath{\mathcal{#1}}}

\providecommand{\keywords}[1] 
{
  \small	
  \textbf{Keywords:} #1
}

\newcommand{\mythanks}[1]{\vspace{1em}\begin{adjustwidth}{1cm}{1cm}\footnotesize{\emph{#1}}\end{adjustwidth}\vspace{1em}}

\title{Exact linear reduction for rational dynamical systems}

\author{
  Antonio Jim\'{e}nez-Pastor\footnote{\url{jimenezpastor@lix.polytechnique.fr}, LIX, CNRS, \'Ecole Polytecnique, Institute Polytechnique de Paris, Palaiseau, 91120, France.}, 
  Joshua Paul Jacob\footnote{\url{joshuapjacob@berkeley.edu}, University of California, Berkeley, 94720, California, USA.}, 
  Gleb Pogudin\footnote{\url{gleb.pogudin@polytechnique.edu}, LIX, CNRS, \'Ecole Polytecnique, Institute Polytechnique de Paris, Palaiseau, 91120, France.}
}
\date{}

\begin{document}








\maketitle

\begin{abstract}
  Detailed dynamical systems models used in life sciences may include dozens or even hundreds of state variables.
  Models of large dimension are not only harder from the numerical perspective (e.g., for parameter estimation or simulation), but it is also becoming challenging to derive mechanistic insights from such models.
  Exact model reduction is a way to address this issue by finding a self-consistent lower-dimensional projection of the corresponding dynamical system.
  A recent algorithm CLUE allows one to construct an exact linear reduction of the smallest possible dimension such that the fixed variables of interest are preserved.
  However, CLUE is restricted to systems with polynomial dynamics.
  Since rational dynamics occurs frequently in the life sciences (e.g., Michaelis-Menten or Hill kinetics), it is desirable to extend CLUE to the models with rational dynamics.
  
  In this paper, we present an extension of CLUE to the case of rational dynamics and demonstrate its applicability on examples from literature.
  Our implementation is available in version 1.5 of CLUE\footnote{\url{https://github.com/pogudingleb/CLUE}}.  
\end{abstract}


\mythanks{%
This work was supported by the Paris Ile-de-France region (via project ``XOR'').
GP was partially supported by NSF grants DMS-1853482, DMS-1760448, and DMS-1853650; and AJP was partially supported by Poul Due Jensen Grant 883901.}

\section{Introduction}\label{sec:introduction}

Dynamical systems modeling is one of the key mathematical tools for describing phenomena in life sciences.
Making such models realistic often requires taking into account a wide range of factors yielding models of high dimension (dozens or hundreds of state variables).
Because of their size, these models are challenging from the numerical standpoint (e.g., parameter estimation) and it may be hard to derive mechanistic insight from them.
One possible workaround is to use \emph{model reduction} techniques that replace a model with a simpler one while preserving, at least approximately, some of the important features of the original model.
For approximate model reduction, many powerful techniques have been developed including singular value decomposition~\cite{antoulas} and time-scale separation~\cite{okino1998}.

A complementary approach is to perform~\emph{exact model reduction}, that is, lower the dimension of the model without introducing approximation errors.
For example, exact linear lumping aims at writing a self-consistent system of differential equations for a set of \emph{macro-variables} in which each macro-variable is a linear combination of the original variables.
The case of the macro-variables being sums of the original variables has been studied for some important classes of biochemical models (see, e.g., ~\cite{Borisov,Conzelmann,Feret}) and for general rational dynamical systems in~\cite{Cardelli2017a,Cardelli2019}.
The latter line of research has culminated in the powerful ERODE software~\cite{Cardelli2017a} which finds the optimal partition of the original variables into macro-variables.

A recent algorithm from~\cite{Ovchinnikov2020} (and implemented in the software CLUE) allows, for a given set of linear forms in the state variables (the \emph{observables}), constructs a linear lumping of the smallest possible dimension such that the observables can be written as combinations of the macro-variables (i.e., the observables are preserved).
Unlike the prior approaches, the macro-variables produced by CLUE are allowed to involve any coefficients (not just zeroes and ones as before).
Thanks to this, one may obtain reductions of significantly lower dimension than it was possible before, see~\cite[Table~1]{Ovchinnikov2020}.

However, unlike ERODE, the algorithm used in CLUE was restricted to the models with polynomial dynamics.
Since rational dynamics appears frequently in life sciences (e.g., via Michaelis-Menten or Hill kinetics), it is desirable to extend CLUE to such models.
The goal of the present paper is to fill this gap.
We design and implement an algorithm computing an optimal linear lumping of a rational dynamical system given a set of observables.
We demonstrate the efficiency and applicability of this algorithm on several models from literature.
Our new algorithm is available in CLUE starting from version 1.5 at~\url{https://github.com/pogudingleb/CLUE}.

In fact, we present two algorithms: one is a relatively direct generalization of the original algorithm from~\cite{Ovchinnikov2020} which works well if there are only few different denominators in the model, and another is a randomized algorithm based on evaluation-interpolation techniques which can efficiently handle models with multiple denominators (see Table~\ref{tab:app2_table} for comparison). 
Both algorithms are based on the same key mathematical fact, going back to~\cite{Li1989}, that linear lumpings are in a bijective correspondence with the invariant subspaces of the Jacobian $J(\bx)$ of the system~\cite[Proposition II.1]{Ovchinnikov2020}.
However, the construction of such invariant subspaces as common invariant subspaces of the coefficients of $J(\bx)$ used in~\cite{Ovchinnikov2020} relies on the polynomiality of the ODE model.
The workaround used in the first algorithm is to make $J(\bx)$ a polynomial matrix by multiplying the common denominator if the entries.
However, if there is more than a couple of denominators, the size of the intermediate expressions becomes prohibitive.
Therefore, in the second algorithm, we used a different approach to replace a non-constant matrix $J(\bx)$ with a collection of constant matrices: we replace $J(\bx)$ by its evaluations at several points which can be exactly computed by means of automatic differentiation without writing down $J(\bx)$ itself.
Since the algorithm is sampling-based, it becomes probabilistic, and we guarantee that the result will be correct with a user-specified probability and also provide a possibility to perform a post-verification of the result.

The rest of the paper is organized as follows. 
Section~\ref{sec:preliminaries} contains the preliminaries on lumping and summarizes the main steps of the algorithms from~\cite{Ovchinnikov2020} for polynomial dynamics.
In Section~\ref{sec:algorithms}, we describe our two algorithms mentioned above.
Section~\ref{sec:performance} described our implementation of the algorithms and reports its runtimes.
In Section~\ref{sec:examples}, we illustrate how our algorithms can be applied to models from the literature.
We conclude in Section~\ref{sec:conclusions}.
Some proofs omitted in the main text are collected in the Appendix.


\section{Preliminaries and prior results}\label{sec:preliminaries}

\subsection{Preliminaries on lumping}

In this paper we study models defined by ODE systems of the form
\begin{equation}\label{equ:ode-system}
    \dot \bx = \bf(\bx),
\end{equation}
where $\bx = (x_1,\ldots,x_n)^T$ are the state variables and $\bf = (f_1,\ldots, f_n)^T$ with $f_1, \ldots, f_n \in \rr(\bx)$ (where $\rr(\bx)$ denotes the set of rational function in $\bx$ with real coefficients).
We will describe lumping and the related notions following~\cite[Section 2]{Ovchinnikov2020} but extending to the case of rational dynamics.

\begin{definition}[{Lumping}]\label{def:lumping}
  Consider a system of the form~\eqref{equ:ode-system}.
  A linear transformation $\by = L \bx$, where $\by = (y_1, \ldots, y_m)^T$ and $L \in \rr^{m\times n}$, is called \emph{lumping} if $\rank(L) = m$ and there exist rational functions $g_1, \ldots, g_m \in \rr(\by)$ such that 
  \[
    \dot{\by} = \bg(\by), \quad \text{where}\quad \bg = (g_1, \ldots, g_m)^T
  \]
  for every solution $\bx$ of~\eqref{equ:ode-system}.
  We call $m$ the \emph{dimension} of the lumping, and we will refer to $\by$ as \emph{macro-variables}.
\end{definition}

In other words, a lumping of dimension $m$ is a linear change of variables from $n$ to $m$ variables such that the new variables 
satisfy a self-contained ODE system.
In the language of differential geometry, a lumping is a linear map $\varphi$ from the state space such that there exists a vector field on the range of the map which is $\varphi$-related to $\mathbf{f}$ (see~\cite[Definition~1.54]{diffgem}).

\begin{example}
Consider the following differential system:
\[\left\{\begin{array}{rcl}
    \dot x_1 &{}={}& \displaystyle\frac{x_2^2 + 4x_2x_3 + 4x_3^2}{x_1^3 - x_2 - 2x_3},\\
    \dot x_2 &{}={}& \displaystyle\frac{4x_3 - 2x_1}{x_1 + x_2 + 2x_3},\\
    \dot x_3 &{}={}& \displaystyle\frac{x_1 + x_2}{x_1 + x_2 + 2x_3},
\end{array}\right.\]
 Then the matrix 
 \[\begin{pmatrix}1 & 0 & 0\\ 0 & 1 & 2\end{pmatrix},\]
 is a lumping of dimension 2, since:
 \[\left\{\begin{array}{rcl}
    \dot y_1 &{}={}& \dot x_1 = \displaystyle\frac{(x_2 + 2x_3)^2}{x_1^3 - x_2 - 2x_3} = \frac{y_2^2}{y_1^3 - y_2}\\
    \dot y_2 &{}={}& \dot x_2 + 2 \dot x_3 = \displaystyle\frac{2x_2 + 4x_3}{x_1+x_2+2x_3} = \frac{2y_2}{y_1 + y_2}.
 \end{array}\right.\]
\end{example}

One way to force a lumping to keep the information of interest is to fix a set of observables to be preserved. 
This leads to the notion of \emph{constrained lumping}.

\begin{definition}[Constrained lumping]
  Let $\bx_{obs}$ be a vector of linear forms in $\bx$ (i.e., there is a matrix $M\in \rr^{p\times n}$ with 
  $\bx_{obs} = M\bx$). 
  We say that a lumping $L$ of $\dot \bx = \bf(\bx)$ is a \emph{constrained lumping} with observables $\bx_{obs}$ if
  each entry of $\bx_{obs}$ can be expressed as a linear combination of $\by = L\bx$.
\end{definition}

\begin{remark}[On the constrained and partition-based lumpings]\label{rem:constrained_partition}
  Software ERODE~\cite{Cardelli2017b} can efficiently produce the minimal (in the sense of the dimension) lumping, in which the macro-variables correspond to a partition of the state variables. This means that $y_1 = \sum_{i \in S_1}x_i, \;\ldots\;, y_m = \sum_{i \in S_m} x_i$ such that $\{1, \ldots, n\} = S_1 \bigsqcup S_2 \bigsqcup \ldots\bigsqcup S_m$.
  
  In this case, we always have $y_1 + \ldots + y_m = x_1 + \ldots + x_n$, hence a lumping found by ERODE will be always a constrained lumping with the sum of the state variables being the only observable.
  Therefore, an algorithm (like the one proposed in this paper) for computing a constrained linear lumping of the smallest possible dimension will be always able to find either the lumping found by ERODE or even lump it further.
  To be fair, we would like to point out that ERODE is typically faster than our algorithm.
\end{remark}

\begin{remark}[Lumping via polynomialization]
    It is known that, by introducing new variables, a rational dynamical system can always be embedded into a polynomial one.
    Therefore, a natural approach to finding constrained linear lumpings for~\eqref{equ:ode-system} would be to combine such an embedding with any algorithm applicable to polynomial dynamics (e.g., CLUE).
    However, it has been demonstrated in~\cite[p. 149]{Cardelli2019} that such an embedding not only increases the dimension of the ambient space but also may miss some of the reductions of original model.
\end{remark}

\subsection{Overview of the CLUE algorithm for polynomial dynamics~\cite{Ovchinnikov2020}}\label{subsec:overview}

The algorithm for computing a constrained lumping of the minimal dimension for a given system $\dot{\bx} = \bf(\bx)$ with polynomial right-hand side and observables $\bx_{obs}$, designed and implemented in~\cite{Ovchinnikov2020}, can be summarized as follows:

\begin{algorithm}[H]
  \caption{Finding constrained linear lumping (polynomial case)} \label{alg:poly}
  \begin{description}
    \item[Input:] a \emph{polynomial} ODE system $\dot{\bx} = \bf(\bx)$ of dimension $n$; a list of observables $\bx_{obs} = A\bx$ for $A \in \rr^{s \times n}$.
    \item[Output:] minimal lumping $L$ containing $\bx_{obs}$.
  \end{description}
  
  \begin{enumerate}[label=\textbf{(\arabic*)}]
      \item\label{step:Jac_comp} Compute the Jacobian matrix $J(\bx)$ of $\bf$;
    \item\label{step:Jac_decomp} Write $J(\bx)$ as $J_1m_1 + \ldots + J_Nm_N$, where $m_1, \ldots, m_N$ are distinct monomials in $\bx$ and $J_1, \ldots, J_N$ are constant matrices;
    \item\label{step:saturate} Compute the minimal subspace $V$ of the space of linear forms in $\bx$ containing $\bx_{obs}$ and invariant under $J_1, \ldots, J_N$ (using \cite[Alg.~3 or~4]{Ovchinnikov2020});
    \item\label{step:return} Return matrix $L$ with rows being basis vectors of $V$.
  \end{enumerate}
\end{algorithm}

This algorithm is based on the criterion from~\cite{Li1989}, which states that a matrix $L$ is a lumping for the system $\dot \bx = \bf(\bx)$ if an only if the row space of $L$ is invariant under $J(\bx)$ for every value of $\bx$.
In order to use this criterion, it was shown in~\cite[Supplementary Materials, Lemma I.1]{Ovchinnikov2020}, that $L$ is a lumping for $\dot x = \bf(\bx)$ if and only if the row space of $L$ is invariant under $J_i$ for $1\leqslant i \leqslant N$ (as defined in Step~\ref{step:Jac_decomp} of Algorithm~\ref{alg:poly}). 
This reduces the problem of finding the lumping to the one solved in Step~\ref{step:saturate} of Algorithm~\ref{alg:poly}.


\section{Algorithm for rational dynamical systems}\label{sec:algorithms}

In this section, we will use the following ``finite'' version of the Jacobian-based criterion from~\cite{Li1989} which is proved in the Appendix.

\begin{restatable}{lemma}{invariant}\label{lem:gen_invariant}
    Consider the rational dynamical system $\dot \bx = \bf(\bx)$ and $J(\bx)$ the Jacobian matrix of $\bf(\bx)$. 
    Let $\goth{B}$ be any set of matrices spanning the vector space $\langle J(\bx) \mid \bx \in \rr^n \text{ and } J(\bx) \text{ is well-defined}\rangle$. 
    Then $L \in \rr^{r\times n}$ is a lumping if and only if the row space of $L$ is invariant with respect to all $J \in \goth{B}$.
\end{restatable}

\begin{remark}[Lemma~\ref{lem:gen_invariant} and the polynomial case]
In the context of Algorithm~\ref{alg:poly}, each value of $J(\mathbf{x})$ is a linear combination of $J_1, \ldots, J_N$, so $\goth{B}$ can be taken to be $\{J_1, \ldots, J_N\}$.
Thus Lemma~\ref{lem:gen_invariant} implies the correctness of Algorithm~\ref{alg:poly}.
\end{remark}

\subsection{Straightforward extension of Algorithm~\ref{alg:poly}}\label{sec:direct}

In the case of rational dynamics, the Jacobian matrix $J(\bx)$ has only rational function entries, so we can compute the common denominator $q(\bx)$ such that $J(\bx) = \frac{B(\bx)}{q(\bx)}$, where $B(\bx)$ is a matrix with polynomial entries. 
Let $B(\bx) = B_1 m_1 + \ldots + B_N m_N$ be a decomposition where $B_1, \ldots, B_N$ are constant matrices and $m_1, \ldots, m_N$ are distinct monomials appearing in $B(\bx)$ (compare with Step~\ref{step:Jac_decomp} of Algorithm~\ref{alg:poly}).
Then, for each value of $\bx$ not annihilating $q$, $J(\bx)$ is a linear combination of $B_1, \ldots, B_N$, so one can take $\goth{B} = \{B_1, \ldots, B_N\}$ in Lemma~\ref{lem:gen_invariant}.
This yields the following algorithm for rational case:

\begin{algorithm}[H]
  \caption{Finding constrained linear lumping (rational case)} \label{alg:build_matrices_rat}
  \begin{description}
    \item[Input:] a \emph{rational} ODE system $\dot{\bx} = \bf(\bx)$ of dimension $n$; a list of observables $\bx_{obs} = A\bx$ for $A \in \rr^{s \times n}$.
    \item[Output:] minimal lumping $L$ containing $\bx_{obs}$.
  \end{description}
  
  \begin{enumerate}[label=\textbf{(\arabic*)}]
      \item Compute $J(\mathbf{x})$, the Jacobian of $f(\mathbf{x})$;
      \item[\textbf{(2.1)}] Compute $p(\mathbf{x})$, the common denominator of the entries of $J(\mathbf{x})$
      \item[\textbf{(2.2)}] Set $B(\mathbf{x}) := p(\mathbf{x}) \cdot J(\mathbf{x})$
      \item[\textbf{(2.3)}] Write $B(\bx)$ as $B_1 m_1 + \ldots + B_N m_N$, where $m_1, \ldots, m_N$ are distinct monomials and $B_1, \ldots, B_N$ are constant matrices.
      \item[\textbf{(3)}] Compute the minimal subspace $V$ of the space of linear forms in $\bx$ containing $\bx_{obs}$ and invariant under $B_1, \ldots, B_N$ (using \cite[Alg.~3 or~4]{Ovchinnikov2020});
      \item[\textbf{(4)}] Return matrix $L$ with rows being basis vectors of $V$;
  \end{enumerate}
\end{algorithm}
However, the size of the expressions involved after bringing everything to the common denominator and differentiation can be prohibitively large. 
The following toy example illustrates this phenomenon (for a comparison with the more refined Algorithm~\ref{alg:find_cl}, see Section~\ref{sec:comparison}).

\begin{example}\label{exm:too_big}
    Consider the following differential system:
    \[\dot x = \frac{y - z}{x - y},\quad \dot y = \frac{x + z}{x + y},\quad \dot z = \frac{x+y+z}{z - x - y}.\]
    
    Here is the Jacobian matrix:
    \[J(x,y,z) = \begin{pmatrix}
        -\frac{y-z}{(x-y)^2} & \frac{y-z}{(x+y)^2} & \frac{2z}{(x+y-z)^2}\\
        -\frac{x-z}{(x-y)^2} & -\frac{x+z}{(x+y)^2} & \frac{2z}{(x+y-z)^2}\\
        -\frac{-1}{x-y} & \frac{1}{x+y} & \frac{-2(x+y)}{(x+y-z)^2}\\
    \end{pmatrix}\]
    
    The polynomial $p(\bx)$ from Algorithm~\ref{alg:build_matrices_rat} will be equal to  $(x-y)^2(x+y)^2(x+y-z)^2$. 
    Then the matrix $B(x,y,z)$ will satisfy
    \[J(x,y,z) = \frac{1}{(x-y)^2(x+y)^2(x+y-z)^2}B(x,y,z).\]
    In this example, the matrix $B(x,y,z)$ has all entries of degree 5 that we need to expand which is substantially more complicated than the original expressions.
\end{example}

\subsection{The main algorithm based on evaluation-interpolation}

Our main algorithm is built upon the following observation: we can evaluate $J(\bx)$ efficiently at a given point without writing down the symbolic matrix explicitly (which would be quite large even in small examples such as Example~\ref{exm:too_big}), using automatic differentiation techniques. 
Then we can use sufficiently many such evaluations to span the space~$\langle J(\bx) \mid \bx \in \rr^n \text{ and } J(\bx) \text{ is well-defined}\rangle$ from Lemma~\ref{lem:gen_invariant}.
The resulting Algorithm~\ref{alg:find_cl} is shown below. It relies on Algorithm~\ref{alg:sampling} for generating ``sufficiently many'' evaluations which we present in Section~\ref{sec:sampling}.

\begin{algorithm}[H]
  \caption{Finding constrained linear lumping (probabilistic)}\label{alg:find_cl}
  \begin{description}
    \item[Input:] a rational ODE system $\dot{\bx} = \bf(\bx)$ of dimension $n$; a list of observables $\bx_{obs} = A\bx$ for $A \in \rr^{s \times n}$; and a real number $\varepsilon \in (0, 1)$.
    \item[Output:] minimal lumping $L$ containing $\bx_{obs}$.\\ The result is correct with probability at least $1 - \varepsilon$.
  \end{description}
  
  \begin{enumerate}[label=\textbf{(\arabic*)}]
      \item Compute $J(\mathbf{x})$, the Jacobian of $f(\mathbf{x})$.
      \item\label{step:points} Compute points $\bx_1, \ldots, \bx_M \in \mathbb{Q}^n$ such that $J(\bx_1),\ldots,J(\bx_M)$ span $\langle J(\bx) \mid \bx \in \rr^n \text{ and } J(\bx) \text{ is well-defined}\rangle$ with probability at least $1 - \varepsilon$ (using Algorithm~\ref{alg:sampling}).
      \item\label{step:other} Compute the minimal subspace $V$ of linear forms in $\bx$ containing $\bx_{obs}$ and invariant under $J(\bx_1), \ldots, J(\bx_M)$ (using \cite[Alg.~3 or~4]{Ovchinnikov2020});
      \item Return matrix $L$ with rows being basis vectors of $V$;
  \end{enumerate}
\end{algorithm}

\begin{remark}[On probabilities in algorithms]\label{rem:prob}
  The outputs of Algorithm~\ref{alg:find_cl} and~\ref{alg:sampling} are guaranteed to be correct with a user-specified probability.
  This should be understood as follows.
  The algorithm makes some random choices, that is, draws a point from a probability space.
  The specification of the algorithm means that, for the fixed input, the probability of the output being correct (with respect to the probability space above) is at least ${1 - }\varepsilon$.
  
  In the highly unlikely case (the used probability bounds are quite conservative) the computation at Step~\ref{step:points} of Algorithm~\ref{alg:find_cl} was incorrect, the computed space $V$ will be a subspace of the ``true'' $V$, so the resulting matrix will not provide a lumping, and the reduced model returned by the software will be incorrect.
  The returned model can be checked using a direct substitution. 
  This substitution is not performed in our implementation by default, but we offer a method to gain full confidence in the result.
\end{remark}

\begin{proposition}
  Algorithm~\ref{alg:find_cl} is correct, that is the returned matrix $L$ is a minimal lumping with probability at least $1 - \varepsilon$ (see Remark~\ref{rem:prob}).
\end{proposition}

\begin{proof}
    Assume that we have found points $\bx_1,\ldots, \bx_M \in \rr^n$ such that the evaluations $J(\bx_1), \ldots, J(\bx_M)$ span $\langle J(\bx) \mid \bx \in \rr^n \text{ and }J(\bx) \text{ is well-defined}\rangle$, and consider $L$, a matrix with the rows being basis vectors of $V$, computed using~\cite[Algorithm~3 or Algorithm~4]{Ovchinnikov2020}.
    Then it will be a lumping by Lemma~\ref{lem:gen_invariant}.
    
    If that is not the case, then it means that $J(\bx_1), \ldots, J(\bx_M)$ do not span $\langle J(\bx) \mid \bx \in \rr^n \text{ and }J(\bx) \text{ is well-defined}\rangle$. 
    This event only happens with probability at most ${\varepsilon}$. 
    Hence the algorithm is correct with at least probability~${1 - }\varepsilon$.
\end{proof}


\subsection{Generating ``sufficiently many'' evaluations}\label{sec:sampling}

In order to complete Algorithm~\ref{alg:find_cl}, we will present in this section a procedure (Algorithm~\ref{alg:sampling}) for sampling values of $J(\bx)$ spanning the whole $\langle J(\bx) \mid \bx \in \rr^n \text{ and } J(\bx) \text{ is well-defined}\rangle$ with high probability.
The main theoretical tool to achieve the desired probability is the following proposition (proved in the appendix) based on the Schwartz-Zippel lemma~\cite[Proposition~98]{Zippel}.

\begin{restatable}{proposition}{sampling}\label{prop:sampling}
  Let $\bf = (f_1, \ldots, f_n)^T$ be a vector of elements of $\rr(\bx)$, where $\bx = (x_1, \ldots, x_n)^T$, and $\varepsilon \in (0, 1)$ be a real number.
  Let $D_d$ (resp.,  $D_n$) be  an integer such that the degree of the denominator (resp., numerator) of $f_i$ does not exceed $D_d$ (resp., $D_n$) for every $1 \leqslant i \leqslant n$.
  
  Let $J(\bx)$ be the Jacobian matrix of $\bf$ and $\bx_1, \ldots, \bx_m$ be points such that $J(\bx_1), \ldots, J(\bx_m)$ do not span $\langle J(\bx) \mid \bx \in \rr^n \text{ and } J(\bx) \text{ is well-defined}\rangle$.
  Consider a point $\bx_{m + 1}$ with each coordinate being an integer sampled uniformly at random from $\{1, 2, \ldots, N\}$ where 
  \[
    N > \frac{D_n + (2m + 1)D_d}{\varepsilon} + n D_d.
  \]
  Then we have
  \[
    \mathbb{P}[J(\bx_{m + 1}) {\in} \langle J(\bx_1), \ldots, J(\bx_m)\rangle \mid J(\bx_{m + 1}) \text{ is well-defined}] {<} \varepsilon.
  \]
\end{restatable}

Proposition~\ref{prop:sampling} tells us that, when the points are sampled from a large enough range, if a value of $J(\bx)$ belongs to the space spanned by the previous evaluations, then these evaluations span, with high probability, the whole space $\langle J(\bx) \mid \bx \in \rr^n \text{ and } J(\bx) \text{ is well-defined}\rangle$.
This yields the following algorithm.

\begin{algorithm}[H]
  \caption{Sampling the values of the Jacobian}\label{alg:sampling}
  \begin{description}
    \item[Input:] \begin{itemize}
    \item[]
        \item $n$-dimensional vector $\mathbf{f}$ of rational functions in $\bx = (x_1, \ldots, x_n)$;
        \item real number $\varepsilon \in (0, 1)$.
    \end{itemize}
    \item[Output:] Points $\bx_1, \ldots, \bx_M \in \mathbb{Q}^n$ such that
    \begin{equation}\label{equ:spaces_equality}
      \langle J(\bx_i) \mid 1 \leqslant i \leqslant n \rangle = \langle J(\bx) \mid \bx \in \rr^n \text{ and }J(\bx) \text{ is well-defined} \rangle,
    \end{equation}
    with probability at least ${1 - }\varepsilon$, where $J(\bx)$ is the Jacobian matrix of $\mathbf{f}$.
  \end{description}
  
  \begin{enumerate}[label=\textbf{(\arabic*)}]
      \item Compute $D_d$ (resp., $D_n$) as the maximum of the degrees of denominators (resp., numerators) of entries of $\mathbf{f}$, respectively.
      \item $M \gets 0$
      \item\label{step:sampling} Repeat
      \begin{enumerate}
          \item Compute $\bx_{M + 1}$ by sampling each coordinate uniformly at random from $\{1, 2, \ldots, N\}$, where
          \[
            N = \left[ \frac{D_n + (2M + 1)D_d}{\varepsilon} + n D_d \right] + 1.
          \]
          Repeat sampling until none of the denominators of $\bf$ vanishes at $\bx_{M + 1}$.
          \item Compute $J(\bx_{M + 1})$ using automatic differentiation (see Section~\ref{sec:speedup}).
          \item If $J(\bx_{M + 1}) \in \langle J(\bx_i) \mid 1 \leqslant i \leqslant M\rangle$, return $\bx_1, \ldots, \bx_M$.
          \item $M \gets M + 1$
      \end{enumerate}
  \end{enumerate}
\end{algorithm}

\begin{proposition}
  Algorithm~\ref{alg:sampling} is correct, that is, for the returned points $\mathbf{x}_1, \ldots, \mathbf{x}_M$ the relation~\eqref{equ:spaces_equality} holds with probability at least $1 - \varepsilon$ (see Remark~\ref{rem:prob}).
\end{proposition}

\begin{proof}
    Assume that the output condition~\eqref{equ:spaces_equality} for Algorithm~\ref{alg:sampling} does not hold.
    Since the algorithm has returned, the value $J(\bx_{M + 1})$ belonged to the span of $J(\bx_1), \ldots, J(\bx_M)$.
    Then Proposition~\ref{prop:sampling} implies that the probability of this event is at most ${\varepsilon}$, so the output of the algorithm is correct with a probability of at least ${1 - }\varepsilon$.
\end{proof}


\subsection{Improving the efficiency of Algorithm~\ref{alg:find_cl}}\label{sec:speedup}

\paragraph{Evaluating the Jacobian.} 
An attractive feature of Algorithm~\ref{alg:find_cl} is that it does not require a symbolic expression for the Jacobian $J(\bx)$ but only a way to efficiently evaluate it at chosen points.
These evaluations can be preformed efficiently using exact automatic differentiation~\cite{Baydin2018,Elliott2009,Wengert1964,HitchAD}.

More precisely, our implementation uses a forward algorithm for automatic differentiation by computing with the extended dual numbers.
Extended dual numbers are tuples $(a_0, a_1, \ldots, a_n)$ of real numbers (where $n$ is the dimension of the model) with the following arithmetic rules:
\begin{align*}
    (a_0,a_1,\ldots, a_n) + (b_0,b_1,\ldots, b_n) &= (a_0 + b_0, \ldots, a_n + b_n),\\
    (a_0,a_1,\ldots, a_n)(b_0,b_1,\ldots, b_n) &= (a_0b_0, a_1b_0 + a_0b_1,\ldots, a_0b_n + a_n b_0).
\end{align*}
If one evaluates a rational function $f(\bx)$ at the point 
\begin{equation}\label{eq:point}
  \left((x_1, 1, 0, \ldots, 0), (x_2, 0, 1, \ldots, 0), \ldots, (x_n, 0, 0, \ldots, 1)\right),
\end{equation}
the output will be (see~\cite[Section~4]{HitchAD}) precisely the vector
\[
\left(f(\bx), \frac{\partial f}{\partial x_1}(\bx),\ldots, \frac{\partial f}{\partial x_n}(\bx)\right).
\]
Thus, we can evaluate $J(\bx)$ by evaluating $f_1, \ldots, f_n$ at~\eqref{eq:point}.
The complexity of such evaluation is directly related to the cost of evaluating $f_1, \ldots, f_n$. 
If evaluating each of $f_i(\bx)$ costs $A$ operations, then evaluating the whole $J(\bx)$ has cost $\mathcal{O}(n^2A)$. 
There are techniques (based on backward mode automatic differentiation) that could compute these evaluations within $\mathcal{O}(nA)$ operations~\cite{Baur1983}. 
For the moment, we decided to stick to the forward mode due to its simplicity and generalizability.

\paragraph{Selecting starting evaluation points.}
{The sampling strategy for the evaluation points $\bx_1, \ldots, \bx_M$ at Step~\ref{step:sampling} of Algorithm~\ref{alg:sampling} derived from Proposition~\ref{prop:sampling} is used to ensure (with the prescribed probability) that we do not stop sampling too early.
Therefore, before going into Step~\ref{step:sampling}, we can start with choosing several evaluation points $\bx_1, \ldots, \bx_\ell$ in a different way than it is described in Step~\ref{step:sampling}, and the probability of correctness will be preserved.}

More precisely, before starting to sample evaluation points as in Step~\ref{step:sampling}, our implementation of Algorithm~\ref{alg:sampling} proceeds as follows:
\begin{enumerate}
    \item \emph{Sparse points.}
    For every $1\leqslant i \leqslant n$, we do the following.
    Let $\mathbf{e}_1$ be the $i$-th standard basis vector in $\mathbb{R}^n$.
    We consider the point $\mathbf{e}_i$ and analyze the variables that need to be made different from zero so that the Jacobian $J(\bx)$ will be well-defined.
    We do this modifications in $\mathbf{e}_i$ and obtain (typically sparse) evaluation point.
    \item \emph{``Small points''.}
    Then we consider randomly sampled evaluation points but sample them from a small range until we detect a linear dependence between the evaluations of the Jacobian at already sampled points.
\end{enumerate}

These evaluation points yield simpler evaluations of the Jacobian $J(\bx)$ which simplifies further computations.
The optimization has led to significant speedup, see Table~\ref{tab:sparse_points} (the models used in the table are described in Section~\ref{sec:runtimes}).
\begin{table}[H]
    \centering
    \begin{tabular}{|c|c||c|c|}
        \hline
        \textbf{Model} &\textbf{Speedup} & \textbf{Model} &\textbf{Speedup} \\ \hline
        \textit{BIOMD0000000013} & $4.72$ & \textit{Section~\ref{subsec:mmk}} ($n=4$) & $5.31$ \\\hline
        \textit{BIOMD0000000023} & $2.35$ & \textit{Section~\ref{subsec:mmk}} ($n=5$) & $31.45$\\\hline
        \textit{BIOMD0000000033} & $1.92$ & \textit{Section~\ref{subsec:mmk}} ($n=6$)& $>100$\\\hline
        \textit{MODEL1502270000} & $8.98$ & & \\\hline
    \end{tabular}
    \caption{Speedup through a refined choice of the evaluation points}
    \label{tab:sparse_points}
\end{table}


\section{Implementation and performance}\label{sec:performance}
We implement both our algorithms described in Section~\ref{sec:algorithms} in the software CLUE~\cite{Ovchinnikov2020}.
This software is written in Python and before it allowed to compute optimal constrained lumpings for any polynomial dynamical system. 
We have extended the functionality to the case of rational dynamics (starting with version 1.5).
CLUE is available on the following GitHub repository:
\begin{center}\url{https://github.com/pogudingleb/CLUE}\end{center}
and can be installed using \texttt{pip} directly from GitHub using the command line
\begin{center}\texttt{pip install git+https://github.com/pogudingleb/CLUE}\end{center}
We refer to the README and tutorial in the repository for further details on how to use the software.

\subsection{Performance of Algorithm~\ref{alg:find_cl}}\label{sec:runtimes}

In this section, we report the runtimes for our implementation of our main algorithm, Algorithm~\ref{alg:find_cl}.
The timings reported below were measured on a laptop with Intel i7-9850H, 16GB RAM, and Python 3.8.10. 
The runtimes are average through 5 independent executions on each model. 
We have measured runtimes for two sets of models
\begin{itemize}
    \item Several models with rational dynamics from the BioModels database~\cite{BioModels2020} in Table~\ref{tab:bio_examples}. 
    For each model, the table contains the name in the database and a reference to the related paper.
    {The observables in Table~\ref{tab:bio_examples} were chosen as the states yielding nontrivial reductions, further analysis of these reductions is a question of future research.}
    \item The models we use as examples in this paper (see Section~\ref{sec:examples}) in Table~\ref{tab:examples_timing}.
\end{itemize}
{In these examples, we have fixed the probability bound $\varepsilon = 0.01$. Changing this value to $0.001$ may affect the timings measured, although its effect very moderate (less than 3\% in all models).}
In both tables, we also include the observable used in the lumping and the change of the dimension (in the format $\text{\em before} \to \text{\em after}$).

\begin{table}[H]
    \centering
    \begin{adjustbox}{center}
    \begin{tabular}{|>{\footnotesize\em}c|c|>{\footnotesize\em}l|l|c|}
        \hline
        {\normalsize\em\textbf{Name}} & \textbf{Reference} & \multicolumn{1}{c|}{\textbf{Obs.}} & \multicolumn{1}{c|}{\textbf{Size}} & \textbf{Time  (min.)} \\ \hline
        BIOMD0000000013 & \cite{Poolman2004} & x\_CO2 & $28 \rightarrow 25$ & $2.19$\\  \hline
        BIOMD0000000023 & \cite{Rohwer2001} & Fru & $13 \rightarrow 11$ & $0.04$\\  \hline
        BIOMD0000000113 & \cite{Dupont1992} & Y & $20 \rightarrow 12$ & $0.01$\\  \hline
        BIOMD0000000182 & \cite{Neves2008} & AC\_cyto\_mem & $45 \rightarrow 14$ & $44.1$\\  \hline
        BIOMD0000000313 & \cite{Raia2010} & IL13\_DecoyR & $35 \rightarrow 5$ & $0.09$\\  \hline
        BIOMD0000000448 & \cite{Brnnmark2013} & mTORC1a & $67 \rightarrow 48$ & $1.37$\\  \hline
        BIOMD0000000526 & \cite{Kallenberger2014} & DISC & $32 \rightarrow 19$ & $0.1$\\  \hline
        MODEL1502270000 & \cite{Weisse2015} & rmr & $46 \rightarrow 45$ & $56.44$\\  \hline
    \end{tabular}
    \end{adjustbox}
    \caption{Execution time for Algorithm~\ref{alg:find_cl} for models from BioModels}
    \label{tab:bio_examples}
\end{table}

\begin{table}[H]
    \centering
    \begin{adjustbox}{center}
    \begin{tabular}{|>{\footnotesize\em}c|>{\footnotesize}l|>{\footnotesize\em}l|c|c|}
        \hline
        {\normalsize\em\textbf{Example}} & \multicolumn{1}{c|}{\textbf{Obs.}} & \multicolumn{1}{c|}{\textbf{Size}} & \textbf{Time (min.)} \\ \hline
        Section~\ref{subsec:nerve_growth} &  freeEGFReceptor & $80 \rightarrow 5$ & $0.25$\\  \hline
        Section~\ref{subsec:mmk}, $n=6$ &  $x_1$ & $\ 6\rightarrow 2$ & $0.02$\\ \hline
        Section~\ref{subsec:mmk}, $n=7$ &  $x_1$ & $\ 7\rightarrow 2$ & $0.04$\\ \hline
        Section~\ref{subsec:mmk}, $n=8$ &  $x_1$ & $\ 8\rightarrow 2$ & $0.17$\\ \hline
        Section~\ref{subsec:mmk}, $n=9$ &  $x_1$ & $\ 9\rightarrow 2$ & $0.33$\\ \hline
        Section~\ref{subsec:mmk}, $n=10$ &  $x_1$ & $10\rightarrow 2$ & $3.74$\\ \hline
    \end{tabular}
    \end{adjustbox}
    \caption{Execution time for Algorithm~\ref{alg:find_cl} for the examples from Section~\ref{sec:examples}}
    \label{tab:examples_timing}
\end{table}

{For all the models in Table~\ref{tab:bio_examples}, the optimal reduction found by the algorithm is not of the type that could be found by ERODE, so the reduction by ERODE would be of higher dimension (but faster to compute).
For example, the reduction found by ERODE for BIOMD0000000182 would be of dimension 33 while we reduce further to 14. 
We leave a more detailed comparison of CLUE and ERODE for future research.}

From the timings above one can see that the complexity of the model for our algorithm is not solely determined by its order but rather by its structure: compare, for example, BIOMD0000000013 and BIOMD0000000448.

We have analyzed the breakdown of the total runtimes for these examples and observed that the most time-consuming part is Algorithm~\ref{alg:sampling} while reading the models and computing the actual lumping typically take less than $1\%$ of the full computation time.
Since the sampled matrices depend only on the model, not on the observables, we enhanced CLUE with caching of the matrices so that one can reuse them if wants to check several different sets of observables for a single model. 
In this case, every subsequent computation will be much faster ($\sim 100$ times) than reported in the table.

\subsection{Comparing Algorithm~\ref{alg:build_matrices_rat} and Algorithm~\ref{alg:find_cl}}\label{sec:comparison}

In our implementation, both Algorithms~\ref{alg:build_matrices_rat} and~\ref{alg:find_cl} are available.
For rational dynamical systems, the latter (probabilistic) is used by default.
For the polynomial systems, the original algorithm from CLUE is used.

Table~\ref{tab:app2_table} contains the runtimes of Algorithms~\ref{alg:build_matrices_rat} and~\ref{alg:find_cl} for several biological models and, for each model, we count the number of distinct denominators appearing in the right-hand side of the ODE system. 
We separate the polynomial systems that have been taken from the paper~\cite{Ovchinnikov2020} from the rational systems we studied in this paper.
Note that, for the polynomial systems, Algorithm~\ref{alg:find_cl} is essentially the original algorithm from CLUE.
One can observe that, if there is only a couple of denominators, then Algorithm~\ref{alg:build_matrices_rat} may be much faster but once the number of denominators grows, Algorithm~\ref{alg:find_cl} outperforms it substantially.
In the future, it would be interesting to determine an algorithm to use on the fly.

\begin{table}[H]
    \centering
    \begin{adjustbox}{center}
    \begin{tabular}{|>{\footnotesize\em}c|c|c|c|c|}
        \hline
        \em\normalsize\textbf{Polynomial models} & \textbf{Reference} & \textbf{\# denoms} & \textbf{Algorithm~\ref{alg:build_matrices_rat}} & \textbf{Algorithm~\ref{alg:find_cl}} \\ \hline
        Barua & \cite{Barua2009} & $0$ & $1.1913$ & $>30$ \\  \hline
        OrderedPhosphorylation &  \cite{Borisov2008} & $0$ & $0.0131$ & $>30$ \\  \hline
        MODEL1001150000 & \cite{Pepke2010} & $0$ & $0.0116$ & $>30$ \\  \hline
        MODEL8262229752 & \cite{Li2006} & $0$ & $0.0005$ & $0.0262$ \\  \hline
        fceri\_fi & \cite{Faeder2003} & $0$ & $0.0494$ & $>30$ \\  \hline
        ProteinPhosphorylation ($7$) & \cite{Sneddon2010} & $0$ & $3.433$ & $>30$ \\  \hline
        \hline
        \em\normalsize\textbf{Rational models} & \multicolumn{4}{c|}{} \\\hline
        BIOMD0000000526 & \cite{Kallenberger2014} & $1$ & $0.0222$ & $0.0972$ \\ \hline
        BIOMD0000000448 & \cite{Brnnmark2013} & $2$ & $0.0387$ & $1.3719$ \\ \hline
        BIOMD0000000313 & \cite{Raia2010} & $2$ & $0.0111$ & $0.0881$ \\ \hline
        BIOMD0000000113 & \cite{Dupont1992} & $4$ & $0.4914$ & $0.0073$ \\ \hline
        BIOMD0000000013 & \cite{Poolman2004} & $10$ & $>30$ & $2.1906$ \\ \hline
        BIOMD0000000023 & \cite{Rohwer2001} & $11$ & $>30$ & $0.0435$ \\ \hline
        BIOMD0000000033 & \cite{NGF2004} & $22$ & $>30$ & $2.3012$ \\ \hline

    \end{tabular}
    \end{adjustbox}
    \caption{Execution times for Algorithms~\ref{alg:build_matrices_rat} and~\ref{alg:sampling} (in minutes)}
    \label{tab:app2_table}
\end{table}


\section{Examples}\label{sec:examples}

\subsection{Michaelis-Menten kinetics with competing substrates}\label{subsec:mmk}

Consider an enzymatic reaction with a single enzyme $E$ and multiple competing substrates $S_1, \ldots, S_n$ (and the corresponding products $P_1, \ldots, P_n$).
The reaction can be described by the following chemical reaction network
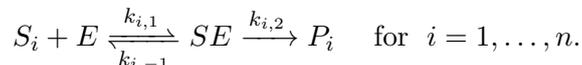
\[
\ce{$S_i + E$ <=>[$k_{i, 1}$][$k_{i, -1}$] $SE$ ->[$k_{i, 2}$] $P_i$}\quad \text{ for }\;i = 1, \ldots, n.
\]
Then quasi-steady-state approximation can be used to deduce a system of ODEs for the substrate concentrations only.
For $n$ competing substrates, using the general approach due to~\cite{ChouTalalay} (for competing substrates, see also~\cite[Section 3]{Schnell2000}), we obtain the following ODE system
\begin{equation}\label{eq:competingMM}
\dot{x}_i = \frac{a_i x_i}{1 + \sum\limits_{j = 1}^n \frac{x_j}{K_j}} \quad \text{ for } i =1, \ldots, n,
\end{equation}
where $x_i$ is the concentration of $S_i$, $K_i = \frac{k_{i, -1} + k_{i, 2}}{k_{i, 1}}$, $a_i = \frac{k_{i, 2}E_0}{K_i}$, and $E_0$ is the total enzyme concentration.
Assume that we are interested in the dynamics of a particular substrate, say $S_1$.
We will now analyze possible constrained lumpings of the system~\eqref{eq:competingMM} depending on the relations between the parameters ($K_i$'s and $a_i$'s).
If $a_i$'s are arbitrary distinct numbers or distinct symbolic parameters, our algorithm shows that there are no nontrivial reductions. 
However, if some of $a_i$'s are equal, the situation becomes more interesting.

\begin{itemize}
    \item \emph{Simplest case: $a_2 = a_3 = \ldots = a_n$.}
    In this case, independently from $n$ (checked for $n \leqslant 10$), the algorithm produces the system
    \[
    \begin{cases}
      y_1' = \frac{a_1 y_1}{1 + \frac{y_1}{K_1} + y_2},\\
      y_2' = \frac{a_2 y_2}{1 + \frac{y_1}{K_1} + y_2},
    \end{cases}
    \quad \text{ where } \quad
    \begin{cases}
      y_1 = x_1,\\
      y_2 = \sum\limits_{i = 2}^n \frac{y_i}{K_i}.
    \end{cases}
    \]
    \item \emph{More general case: some of $a_i$'s are equal.}
    For example, if $n = 6$ and we have $a_2 = a_3$ and $a_4 = a_5 = a_6$, the optimal reduction produced by the algorithm will be
    \[
      \begin{cases}
          y_1' = \frac{a_1 y_1}{1 + \frac{y_1}{K_1} + y_2 + y_3},\\
          y_2' = \frac{a_2 y_2}{1 + \frac{y_1}{K_1} + y_2 + y_3},\\
          y_3' = \frac{a_4 y_3}{1 + \frac{y_1}{K_1} + y_2 + y_3},
      \end{cases}
      \quad \text{ where } \quad
      \begin{cases}
         y_1 = x_1,\\
         y_2 = \frac{x_2}{K_2} + \frac{x_3}{K_3},\\
         y_3 = \frac{x_4}{K_4} + \frac{x_5}{K_5} + \frac{x_6}{K_6}.
      \end{cases}
    \]
    More generally, if several of $a_i$'s are equal, the corresponding $x_i$'s are lumped together with the coefficients $\frac{1}{K_i}$, and the reduced model defines again the Michaelis-Menten kinetics for competing substrates.
    Interestingly, the substrates can be lumped together if the corresponding $a_i$'s are equal but not all the $k_{i, 1}, k_{i, -1}, k_{i, 2}$ (cf. the scaling transformation in~\cite[p. 161]{Schnell2000}).
\end{itemize}
The reduction above cannot be found by ERODE~\cite{Cardelli2017b} unless the $K_i = K_j$ whenever $a_i = a_j$ since the coefficients are not only ones and zeros. 
Note also that we treat all the parameters symbolically (instead simulating them as states with zero derivatives) so that they can appear in the coefficients of the lumping.
To the best of our knowledge, CLUE is the only lumping software with this feature.

Table~\ref{tab:examples_timing} in Section~\ref{sec:performance} reports the runtimes for different values of $n$.


\subsection{Nerve growth factor signaling}\label{subsec:nerve_growth}

Motivated by the study of differentiation of neuronal cells, Brown et al~\cite{NGF2004} considered a model describing the actions of nerve growth factor (NGF) and mitogenic epidermal growth factor (EGF) in rat pheochromocytoma (PC12) cells.
In the model, these factors stimulated extracellular regulated kinase (Erk) phosphorylation with distinct dynamical profiles via a network of intermediate signaling proteins, the network is shown in Figure~\ref{fig:EGFR}.
Each intermediate protein (and Erk) is modeled using two species: active and inactive states. 
The resulting model is described by a system of $32$ differential equations with $48$ parameters, the full system can be found in~\cite[Supplementary materials]{NGF2004} or as BIOMD0000000033 in the BioModels database~\cite{BioModels2020}.
The exact reduction of this model has been earlier studied in~\cite[Section 5.4]{TCS_large_scale} using ERODE.

\begin{figure}[H]
    \centering
    \input{EGFR.tikz}
    \caption{A network diagram describing the action of NGF and EGR on Erk (a simplified version of~\cite[Figure 1]{NGF2004})}
    \label{fig:EGFR}
\end{figure}
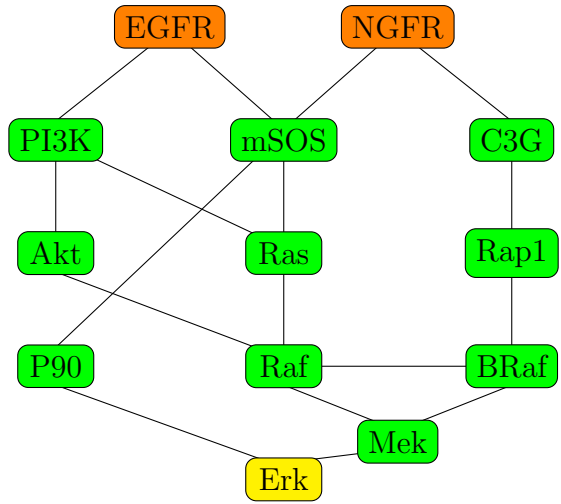

We have applied our algorithm to the model with the observable being the sum of all species (see Remark~\ref{rem:constrained_partition}).
The resulting reduction agrees with the one from~\cite[Section 5.4]{TCS_large_scale}, that is, the macro-variables are the following
\begin{itemize}
    \item concentrations of free and bound EGF and NGF and the corresponding receptors (EGFR and NGFR) remain separate variables;
    \item a single variable with constant dynamics equal to the sum of the concentrations of all other species (intermediate proteins and Erk) is introduced;
    \item only $4$ out of $48$ parameters remain in the reduced model.
\end{itemize}
Therefore, the reduced model is defined by 7 variables and $4$ parameters and captures exactly the dynamics of EGF and NGF.

\section{Conclusions}\label{sec:conclusions}

We have presented the first (to the best of our knowledge) algorithms for finding optimal constrained linear lumping of a rational dynamical system which non trivially extends the existing algorithm for the polynomial case.

While being based on the Jacobian-invariance criterion going back to~\cite{Li1989} and used in the polynomial case~\cite{Ovchinnikov2020}, our main algorithm approaches the key step of~\cite{Ovchinnikov2020}, turning a nonconstant Jacobian matrix into a finite collection of constant matrices, from a different angle, via automatic differentiation and randomized evaluation.
We implement our algorithms, report runtimes for them on a set of benchmarks, and demonstrate how they can be applied to models from the literature.

Directions for future research include extending the algorithm to models involving other functions such as exponential, logarithmic, and trigonometric.
Our current approach is not directly applicable to such functions as one can evaluate them at rational points only approximately, not exactly.

\subsection*{Acknowledgements}

We are grateful to Mirco Tribastone for helpful discussions.
We are grateful to the referees for their helpful suggestions.

\section*{Appendix: Proofs}

\invariant*
\begin{proof}
    Since $\langle J(\bx) \mid \bx \in \rr^n \text{ and } J(\bx) \text{ is well-defined}\rangle$ has finite dimension, then:
    \begin{itemize}
        \item There are points $\bx_1,\ldots,\bx_N \in \rr^n$ such that $J(\bx_1),\ldots, J(\bx_N)$ is a basis.
        \item There are matrices $J_1,\ldots,J_N \in \goth{B}$ that form also a basis.
    \end{itemize}
    
    Hence, there is a invertible matrix $C \in \rr^{N\times N}$ such that (the $n\times n$ matrices are considered as $n^2$-dimensional vectors)
    \[\begin{pmatrix}J_1\\\vdots\\J_N\end{pmatrix} 
    = C\begin{pmatrix}J(\bx_1)\\\vdots\\J(\bx_N)\end{pmatrix}.\]
    
    Assume that $L$ is a lumping, then~\cite[Propsition II.1]{Ovchinnikov2020} implies that there is $A(\bx)$ such that $A(\bx) L = L J(\bx)$. Hence, for the matrices defined by
    \[\begin{pmatrix}A_1\\\vdots\\A_N\end{pmatrix} 
    = C\begin{pmatrix}A(\bx_1)\\\vdots\\A(\bx_N)\end{pmatrix},\]
    we obtain $A_i L = L J_i$. So for all $J\in \goth{B}$ there is $A_J$ such that $A_J L = L J$.
    
    On the other hand, if for all $J\in \goth{B}$ there is $A_J$ such that $A_J L = L J$, then we can check that, if $J(\bx) = \sum_{i=1}^N \alpha_i(\bx) J(\bx_i)$, then
    \[LJ(\bx) = \left[(\alpha_1(\bx),\ldots,\alpha_N(\bx)) C^{-1} \begin{pmatrix} A_{J_1}\\\vdots\\A_{J_N}\end{pmatrix}\right] L,\]
    so the row space of $L$ is invariant under $J(\bx)$ and $L$ is a lumping.
\end{proof}

\sampling*
\begin{proof}
  Define $\mathbf{v}_1,\ldots\mathbf{v}_m \in \rr^{n^2}$ as the vector representations of the matrices $J(\bx_1), \ldots, J(\bx_m)$ and $\bw(x)$ be the corresponding vector representation for the matrix $J(\bx)$ with the entries being rational functions.
  By removing the corresponding evaluation points if necessary, we will further assume that $\mathbf{v}_1, \ldots, \mathbf{v}_m$ are linearly independent.
  The fact that $\langle J(\bx_i) \mid 1 \leqslant i \leqslant m \rangle \neq \langle J(\bx) \mid \bx \in \rr^n \rangle$ implies that
  \[
    \rank E = m + 1,\quad \text{where}\quad E := ( \mathbf{v}_1 \mid \mathbf{v}_2 \mid \ldots \mid \mathbf{v}_m \mid \bw(\bx)).
  \]
  Each maximal minor of $E$ is a linear combination of the entries of $\bw(\bx)$.
  Since $\rank E = m + 1$, at least one of these minors is a nonzero rational function, we write it as $\frac{A(\bx)}{B(\bx)}$.
  Note the degrees of the denominator and numerator of each entry of $\bw(\bx)$ are bounded by $2D_d$ and $D_d + D_n$, respectively.
  Therefore, we have
  \[
    \deg(A(\bx)) \leqslant D_n + (2m + 1) D_d,\qquad \deg(B(\bx)) \leqslant 2(m+1)D_d.
  \]
  Let $q(\bx)$ be the product of the denominators of $f_1, \ldots, f_n$.
  Then $\deg q(\bx) \leqslant n D_d$.
  
  Let $\bx_{m + 1} \in \{1,\ldots N\}^n$ as in the statement of the proposition.
  We would like to find an upper bound for
  \begin{equation}\label{equ:prob}
    \mathbb{P} [J(\bx_{m + 1}) \in \langle J(\bx_1), \ldots, J(\bx_m) \rangle \mid J(\bx_{m + 1}) \text{ is well-defined}].
  \end{equation}
  We write this as
  \[
  \frac{\mathbb{P}[J(\bx_{m + 1}) \in \langle J(\bx_1), \ldots, J(\bx_m) \rangle \text{ and } J(\bx_{m + 1}) \text{ is well-defined}]}{\mathbb{P}[J(\bx_{m + 1}) \text{ is well-defined}]}.
  \]
  For the numerator:
  \begin{multline*}
     \mathbb{P}[J(\bx_{m + 1}) \in \langle J(\bx_1), \ldots, J(\bx_m) \rangle \text{ and } J(\bx_{m + 1}) \text{ is well-defined}] \leqslant\\ \leqslant \mathbb{P} [A(\bx_{m + 1}) = 0] \leqslant \frac{D_n + (2m + 1)D_d}{N},
  \end{multline*}
  where the latter inequality follows from the Schwartz-Zippel lemma~\cite[Proposition~98]{Zippel}.
  Using this lemma again, we bound the probability of the denominator:
  \[
    \mathbb{P}[J(\bx_{m + 1}) \text{ is well-defined}] = \mathbb{P}[q(\bx) \neq 0] \geqslant 1 - \frac{nD_d}{N}.
  \]
  Putting everything together, we obtain
  \[
    \mathbb{P} [J(\bx_{m + 1}) \in \langle J(\bx_1), \ldots, J(\bx_m) \rangle \mid J(\bx_{m + 1}) \text{ is well-defined}] \leqslant \frac{D_n + (2m + 1)D_d}{N - nD_d}.
  \]
  Now a direct computation shows that 
  \[
    \frac{D_n + (2m + 1)D_d}{N - nD_d} < {\varepsilon} \Longleftrightarrow N > \frac{D_n + (2m + 1)D_d}{\varepsilon} + n D_d.
  \]
\end{proof}

\bibliographystyle{splncs04}
\bibliography{main}

\end{document}

%% file: styples.tikzstyles
\tikzstyle{emptycircle}=[fill=white, draw=black, shape=circle, tikzit shape=circle, inner sep=0pt, minimum size=7mm]
\tikzstyle{output}=[fill=white, draw=black, shape=circle, minimum size=3pt]
\tikzstyle{GFR}=[fill={rgb,255: red,255; green,128; blue,0}, draw=black, shape=rectangle, minimum width=1cm, minimum height=0.5cm, rounded corners]
\tikzstyle{Protein}=[fill=green, draw=black, shape=rectangle, minimum width=1cm, minimum height=0.5cm, rounded corners]
\tikzstyle{Erk}=[fill=yellow, draw=black, shape=rectangle, minimum width=1cm, minimum height=0.5cm, rounded corners]

\tikzstyle{newstyle}=[draw=black, ->]

%% file: EGFR.tikz
\begin{tikzpicture}
	\begin{pgfonlayer}{nodelayer}
		\node [style=GFR] (0) at (-7, 9) {EGFR};
		\node [style=GFR] (1) at (-1, 9) {NGFR};
		\node [style=Protein] (2) at (-4, 6) {mSOS};
		\node [style=Protein] (3) at (2, 6) {C3G};
		\node [style=Protein] (4) at (-10, 6) {PI3K};
		\node [style=Protein] (5) at (-10, 3) {Akt};
		\node [style=Protein] (6) at (-4, 3) {Ras};
		\node [style=Protein] (7) at (2, 3) {Rap1};
		\node [style=Protein] (8) at (-10, 0) {P90};
		\node [style=Protein] (9) at (-4, 0) {Raf};
		\node [style=Protein] (10) at (2, 0) {BRaf};
		\node [style=none] (11) at (-7.5, 8.5) {};
		\node [style=none] (12) at (-6.5, 8.5) {};
		\node [style=none] (13) at (-10, 6.5) {};
		\node [style=none] (14) at (-4.25, 6.5) {};
		\node [style=none] (15) at (-3.75, 6.5) {};
		\node [style=none] (16) at (2, 6.5) {};
		\node [style=none] (17) at (2, 5.75) {};
		\node [style=none] (18) at (2, 3.5) {};
		\node [style=none] (19) at (2, 2.5) {};
		\node [style=none] (20) at (2, 0.5) {};
		\node [style=none] (21) at (1, 0) {};
		\node [style=none] (22) at (-3, 0) {};
		\node [style=none] (23) at (-4, 0.5) {};
		\node [style=none] (24) at (-4.75, 0.5) {};
		\node [style=none] (25) at (-10, 2.5) {};
		\node [style=none] (26) at (-10, 3.5) {};
		\node [style=none] (27) at (-4, 5.5) {};
		\node [style=none] (28) at (-4, 3.5) {};
		\node [style=none] (29) at (-4.75, 3.5) {};
		\node [style=none] (30) at (-10, 5.75) {};
		\node [style=none] (31) at (-9.5, 5.75) {};
		\node [style=none] (32) at (-4, 2.5) {};
		\node [style=none] (33) at (-1.5, 8.5) {};
		\node [style=none] (34) at (-0.5, 8.5) {};
		\node [style=none] (35) at (-4.75, 5.5) {};
		\node [style=Protein] (37) at (-1, -2) {Mek};
		\node [style=none] (38) at (-4, -0.5) {};
		\node [style=none] (39) at (2, -0.5) {};
		\node [style=none] (40) at (-0.5, -1.5) {};
		\node [style=none] (41) at (-1.5, -1.5) {};
		\node [style=none] (42) at (-10, 0.5) {};
		\node [style=Erk] (43) at (-4, -3) {Erk};
		\node [style=none] (44) at (-1.5, -2.25) {};
		\node [style=none] (45) at (-10, -0.5) {};
		\node [style=none] (46) at (-3.5, -2.5) {};
		\node [style=none] (47) at (-4.5, -2.5) {};
	\end{pgfonlayer}
	\begin{pgfonlayer}{edgelayer}
		\draw (11.center) to (13.center);
		\draw (12.center) to (14.center);
		\draw (27.center) to (28.center);
		\draw (17.center) to (18.center);
		\draw (19.center) to (20.center);
		\draw (21.center) to (22.center);
		\draw (24.center) to (25.center);
		\draw (15.center) to (33.center);
		\draw (34.center) to (16.center);
		\draw (31.center) to (29.center);
		\draw (26.center) to (30.center);
		\draw (23.center) to (32.center);
		\draw (38.center) to (41.center);
		\draw (40.center) to (39.center);
		\draw (35.center) to (42.center);
		\draw (45.center) to (47.center);
		\draw (46.center) to (44.center);
	\end{pgfonlayer}
\end{tikzpicture}